\documentclass[12pt,reqno]{article}

\usepackage{amssymb}
\usepackage{amsmath}
\usepackage{amsthm}
\usepackage{amsfonts}
\usepackage{amscd}
\usepackage{graphicx}
\usepackage[usenames]{color}
\usepackage[colorlinks=true,
linkcolor=webgreen,
filecolor=webbrown,
citecolor=webgreen]{hyperref}
\usepackage[linesnumbered,ruled,vlined]{algorithm2e}
\usepackage{xcolor}
\usepackage{mdframed}
\usepackage{enumitem}

\definecolor{webgreen}{rgb}{0,.5,0}
\definecolor{webbrown}{rgb}{.6,0,0}

\def\AND{\, \wedge \,}
\def\Enn{\mathbb{N}}

\usepackage{fullpage}
\usepackage{float}

\usepackage{graphics}
\usepackage{latexsym}
\usepackage{epsf}
\usepackage{breakurl}

\newcommand{\seqnum}[1]{\href{https://oeis.org/#1}{\rm \underline{#1}}}

\makeatletter
\def\@fnsymbol#1{\ensuremath{\ifcase#1\or *\or \mathsection\or \mathparagraph\or \|\or **\or \ddagger\ddagger \or \dagger\dagger \else\@ctrerr\fi}}
\makeatother

\def\suchthat{\ : \ }
\def\add{\text{add}}
\DeclareMathOperator{\pre}{Pref}
\DeclareMathOperator{\statecomplexity}{sc}

\begin{document}

\theoremstyle{plain}
\newtheorem{theorem}{Theorem}
\newtheorem{corollary}[theorem]{Corollary}
\newtheorem{lemma}[theorem]{Lemma}
\newtheorem{proposition}[theorem]{Proposition}
\newtheorem{example}[theorem]{Example}
\newtheorem{conjecture}[theorem]{Conjecture}
\newtheorem{openproblem}{Open Problem}
\theoremstyle{definition}
\newtheorem{definition}[theorem]{Definition}
\theoremstyle{remark}
\newtheorem{remark}[theorem]{Remark}

\title{Complexity of Linear Subsequences of $k$-Automatic Sequences}

\author{Delaram Moradi\footnote{Research supported by NSERC grant RGPIN-2024-03725.}\\
School of Computer Science\\
University of Waterloo\\
Waterloo, ON  N2L 3G1 \\
Canada\\
\href{mailto:delaram.moradi@uwaterloo.ca}{\tt delaram.moradi@uwaterloo.ca}\\
\and
Narad Rampersad\footnote{Research supported by NSERC grants RGPIN-2019-04111 and RGPIN-2025-04076.}\\
Department of Math/Stats\\
University of Winnipeg\\
515 Portage Ave.\\
Winnipeg, MB R3B 2E9\\
Canada\\
\href{mailto:n.rampersad@uwinnipeg.ca}{\tt n.rampersad@uwinnipeg.ca}\\
\and 
Jeffrey Shallit$^*$\\
School of Computer Science\\
University of Waterloo\\
Waterloo, ON  N2L 3G1 \\
Canada\\
\href{mailto:shallit@uwaterloo.ca}{\tt shallit@uwaterloo.ca}}

\maketitle

\begin{abstract}
We construct automata with input(s) in base $k$ recognizing some basic relations and study their number of states. 
We also consider some basic operations on $k$-automatic sequences $(h(i))_{i \geq 0}$ and discuss their state complexity.  We find a relationship between subword complexity of the interior sequence $(h'(i))_{i \geq 0}$ and state complexity of the linear subsequence $(h(ni+c))_{i \geq 0}$.   We resolve a recent question of Zantema and Bosma about linear subsequences of $k$-automatic sequences with input in most-significant-digit-first format. We also discuss the state complexity and runtime complexity of using a reasonable interpretation of B\"uchi arithmetic to actually construct some of the studied automata recognizing relations or carrying out operations on automatic sequences.
\end{abstract}

\section{Introduction}

Determining the number of states in a minimal automaton for a formal language or automatic sequence is an extensively studied area of formal languages and automata theory called {\it state complexity}. The reader can refer to works by Yu et al.~\cite{yu_state_1994}, Gao et al.\ \cite{gao_survey_2016}, and Zantema and Bosma \cite{zantema_complexity_2022} for some examples. This paper includes discussions and results on the state complexity of specific languages and automatic sequences with input in base $k$.

The $k$-automatic sequences for integers $k\geq 2$ is an interesting class of sequences
introduced by Cobham \cite{cobham_base-dependence_1969,cobham_uniform_1972} and has been studied for more than 50 years; see 
\cite{allouche_automatic_2003} for a monograph on automatic sequences.
A sequence $(h(i))_{i \geq 0}$ is said to be {\it $k$-automatic\/} if there is a deterministic finite automaton with output (DFAO) computing the $i$-th term, given the base-$k$ representation of $i$ as input.
A DFAO is 
a $6$-tuple $M = (Q, \Sigma, \delta, q_0, \Delta, \tau)$ where
\begin{itemize} [nosep]
    \item $Q$ is a finite nonempty set of states;
    \item $\Sigma = \Sigma_k = \{0, 1,\ldots, k-1 \}$ is the input alphabet;
    \item $\delta:Q \times \Sigma \rightarrow Q$ is the transition
    function, with domain extended to
    $Q \times \Sigma^*$ in the usual way;
    \item $q_0$ is the initial state;
    \item $\Delta$ is the finite nonempty output alphabet; and
    \item $\tau:Q \rightarrow \Delta$ is
    the output function.
\end{itemize}

Recall that if $x = e_{i-1} \cdots e_0$
is a finite word, then $[x]_k$ denotes
the integer value of $x$ when considered as a base-$k$ expansion; that is, 
$$ [x]_k = \sum_{0 \leq j < i} e_j k^j.$$
Thus, for example, $[101011]_2 = 43$ for the most-significant-digit-first (msd-first) format.

We say a DFAO $M$ {\it generates\/} the sequence
${\bf h} = (h(i))_{i \geq 0}$ if
$\tau(\delta(q_0,x)) = h(i)$ for all $i \geq 0$ and all
words $x \in \Sigma_k^*$ such that $[x]_k = i$.  In particular, $M$ is expected to give the correct result {\it no matter how many leading zeros are appended to the input.} A DFAO with base-$k$ input is called  a $k$-DFAO.

\sloppy
For any
automatic sequence ${\bf h} = (h(i))_{i \geq 0}$ generated by the minimal DFAO $(Q, \Sigma_k, \delta, q_0, \Delta, \tau)$, there is a 
unique associated {\it interior sequence} ${\bf h}' = (h'(i))_{i \geq 0}$ with output values from the states of the minimal DFAO generating $(h(i))_{i \geq 0}$, defined
by taking the minimal DFAO and replacing $\tau$ by the identity map $\iota$ on $Q$.  Here ``unique'' means up to renaming of the letters.
Therefore, $(h(i))_{i \geq 0}$ is the image of
$(h'(i))_{i \geq 0}$
under the coding $q \rightarrow \tau(q)$.

For a DFAO $M=(Q, \Sigma, \delta, q_0, \Delta, \tau)$, let the DFAO  $\hat{M}=(Q, \Sigma_k, \delta, q_0, Q, \iota)$, where $\iota:Q \rightarrow  Q$ is the identity map, be the \textit{intrinsic} DFAO of $M$ and let $\hat{\bf M}$ be the sequence generated by the intrinsic DFAO $\hat{M}$.
The interior sequence $(h'(i))_{i \geq 0}$ of an automatic sequence $(h(i))_{i \geq 0}$ is the sequence generated by the intrinsic DFAO of the minimal DFAO for $(h(i))_{i \geq 0}$.  

A classical result of Cobham \cite[p.~174]{cobham_uniform_1972} states that if $(h(i))_{i \geq 0}$ is $k$-automatic, then so is the linear subsequence $(h(ni+c))_{i \geq 0}$ for
integers $n\geq 1$ and $ c \geq 0$.  In this paper, we are interested in the state complexity of linear subsequences of automatic sequences and related topics. Our results on the state complexity of linear subsequences expand on the paper by Zantema and Bosma \cite{zantema_complexity_2022} and we answer one of their open questions. More details can be found in Section \ref{section:linear-subsequences}.

One can also consider automata where the least significant digit is read first.  For least-significant-digit first (lsd-first) input, Boudet and Comon \cite{boudet_diophantine_1996} and for msd-first input, Wolper and Boigelot \cite{wolper_construction_2000} have provided constructions for base-$2$ automata recognizing arithmetical relations. Their constructions relate to some of our results in Section \ref{section:relations}; however, our results are for all bases $k$ and for specific relations.

At first glance, it might seem that results for the state complexity of problems involving msd-first and lsd-first representations would be closely related.  While this is true for arithmetic, it is not true for subsequences of automatic sequences.  Roughly speaking, this is because a deterministic finite-state transducer can carry out addition and multiplication by constants on lsd-first representations, but not msd-first representations.

Boudet and Comon \cite{boudet_diophantine_1996} and Wolper and Boigelot \cite{wolper_construction_2000} also discussed the size of the automaton for an arbitrary formula in Presburger arithmetic.
Furthermore, Amat et al.~\cite{amat_how_nodate} have developed a tool named {\tt Patapsco} that computes a lower bound for the number of states in the automaton for a given Presburger arithmetic formula; this tool computes the lower bound without actually constructing the automaton.
 
Another motivation for studying the topics in this paper is their appearance in expressions in B\"uchi arithmetic, which was implemented (for example) in the \texttt{Walnut} software system \cite{mousavi_automatic_2021} for solving problems in combinatorics on words. 
 Users of this system want to know how long some of the basic expressions can take to evaluate, and how big the resulting automata are. In this paper, we consider the automata for specific relations and automatic sequences. Here we have to examine different considerations than simply state complexity (which measures the size of the {\it smallest possible\/} finite automaton for a language or sequence), because we need to know how large the intermediate automata can be in the constructions implied by B\"uchi arithmetic.   In particular, we need to consider (a) how an expression in B\"uchi arithmetic is translated to an automaton, (b) what are the sizes of the intermediate automata created in this translation, and (c) what the total running time of the procedure is.

In Section \ref{section:background}, we introduce the required background. The rest of the paper is dedicated to discussing our results and their relevance to previously known results. In Section \ref{section:relations}, we discuss results regarding the state complexity of recognizing additive and multiplicative relations with base-$k$ input.
In Section \ref{section:linear-subsequences}, we study the state complexity of linear subsequences of a $k$-automatic sequence, and find a novel relationship between this state complexity and subword complexity.  We also solve an open problem of Zantema and Bosma.   This section contains our main novel results.
In Section \ref{section:Buchi}, we analyze the runtime complexity of constructing some of the automata studied in previous sections using an interpretation of B\"uchi arithmetic. Finally, in Section \ref{section:open-problems}, we state some open problems.

\section{Background} \label{section:background}

For the remainder of this paper, we assume that the inputs to automata are represented in base $k$ and the automatic sequences are $k$-automatic for some constant integer $k\geq 2$. Furthermore, we use $\Sigma_k$ for the set $\{0, \ldots, k-1\}$.

If $x = a_1 \cdots a_i$ and $y = b_1 \cdots b_i$ are words of the same length over the alphabet $\Sigma_k$, 
then by $x \times y$ we mean the word $[a_1, b_1] \cdots [a_i, b_i]$ over the alphabet
$\Sigma_k \times \Sigma_k$, and similarly for more than two inputs.  If the automata we discuss have multiple inputs, we assume they are encoded in parallel in this fashion.

We can specify the order the input is read:  \textit{lsd-first} stands for least significant digit first and \textit{msd-first} stands for most significant digit first. We design our msd-first (resp., lsd-first) automata such that they have the correct behaviour or output no matter how many leading zeros (resp., trailing zeros) there are in the input.

In this paper, we occasionally refer to $O(\log i)$ where $i$ can be $0$ or $1$. In these cases, we adopt the usual convention that $O(\log i) = O(1)$ for $i \in \{0,1 \}$.

The {\it canonical base-$k$ representation\/} of an integer $i$ is the word $x$, without leading zeros, such that $[x]_k = i$.  Furthermore, we use
the notation $(i)_k$ for the canonical representation of $i$ in base $k$.   The canonical representation of $0$ is $\lambda$, the empty word. Sometimes, if the context is clear, we just use $[x]$ instead of $[x]_k$ and $(i)$ instead of $(i)_k$.  We say a word $z$ is a {\it valid right extension\/} of a word $x$
if $x$ is a prefix of $z$.

In addition to DFAOs, we also need the usual notion of {\it deterministic finite automaton} (DFA), which is similar to that of DFAO.  The difference is that
the output alphabet and output function are replaced by a set of
accepting states $A$; an input $x$ is accepted if $\delta(q_0, x) \in A$
and otherwise it is rejected.

Usually, the transition function $\delta$ is taken to be a total
function from $Q \times \Sigma$ to $Q$, but in this paper, we allow
it to be a partial function in order to gracefully handle dead states. 
A state $q$ of a DFA is called {\it dead} if it is not
possible to reach any accepting state from $q$ by a 
(possibly empty) path.  A minimal DFA for a language $L$ has at most one dead state.  By convention, we allow dead states to not be counted or displayed in this paper.

Finally, we also need the notion of {\it nondeterministic finite
automaton} (NFA), which is like that of a DFA, except that the transition
function $\delta: Q \times \Sigma \rightarrow 2^Q$ has set-valued output. Furthermore, the transition function can be extended to domain $2^Q \times \Sigma$. Here an input is accepted if
$\delta(q_0,x)$ contains some element in $A$.

The state complexity of a formal language is defined as the number of states in the minimal DFA recognizing it and we define he state complexity of an automatic sequence as the number of states in the minimal DFAO generating it.

Let $\bf x$ be any infinite word. The number of distinct length-$n$ subwords present in ${\bf x}$ is denoted by the subword complexity function $\rho_{\bf x}(n)$. Here by a {\it subword} we mean a contiguous block of symbols within another word; this concept is also known as a {\it factor}.

We now recall a known result about the subword complexity of automatic sequences.
\sloppy
\begin{theorem} \label{theorem:subword-complexity}
Let ${\bf x}$ be an automatic sequence generated by an $m$-state DFAO $M = (Q, \Sigma, \delta, q_0, \Delta, \tau)$ with msd-first input and let $\hat{\bf M}$ be the sequence generated by the intrinsic DFAO
$\hat{M}=(Q, \Sigma, \delta, q_0, Q, \iota)$, where $\iota$ is the identity map.
Then $\rho_{\bf x} (n) \leq k n  \rho_{\hat{\bf M}} (2) \leq kn m^2$ for all $n \geq 1$. 
\end{theorem}

\begin{proof}
See \cite[Theorem 10.3.1]{allouche_automatic_2003}.  We point out that the statement of the theorem in \cite{allouche_automatic_2003} only gives the inequality $\rho_{\bf x} (n) \leq kn m^2$, but the proof  actually presented there provides the other inequality.
\end{proof}

\section{Recognizing Relations} \label{section:relations}

We say a DFA recognizes a relation (such as the addition relation $i + j = l$) if it takes the values of the variables in parallel, represented in some base $k$, 
and accepts if and only if the relation holds.  In this section we review some of the basic arithmetic operations, what is known about their state complexity, and some new results.

\subsection{Addition}

We start by reviewing a well-known construction for recognizing the addition relation.
\begin{theorem} \label{theorem:addition-FA}
For both lsd-first and msd-first input, there is a $2$-state automaton that recognizes the relation
$[x]_k + [y]_k = [z]_k$.
\end{theorem}

\begin{proof}
    The automaton for msd-first base $2$ and lsd-first base $2$ input can be created by the methods explained, respectively, by Wolper and Boigelot \cite{wolper_construction_2000} and Boudet and Comon \cite{boudet_diophantine_1996}. Furthermore, the automaton for msd-first base $2$ can be found, for example, in Hodgson \cite{hodgson_decidabilite_1983} and the automaton for msd-first base $k$ can be found, for example, in Waxweiler \cite[Proposition 4.4.2]{waxweiler_caractere_2009}; the msd-first base $k$ construction can be easily modified to obtain the lsd-first base $k$ construction.
    
    We only need to do the addition while reading the input and keep track of the carry ($0$ or $1$) by $2$ states.
    The automaton for addition in msd-format is $M=(Q, \Sigma_k, \delta, q_0, A)$ where 
\begin{align*}
        Q &= \{q_0, q_1\} ,\\
        \delta(q, [a, b, c]) &=
                \begin{cases}
                q_0, &  \text{if $q=q_0, c=a+b$}; \\
                q_1, &  \text{if $q=q_0, c=a+b+1$}; \\
                q_1, &  \text{if $q=q_1, c=a+b-k+1$}; \\
                q_0, &  \text{if $q=q_1, c=a+b-k$}; \\
            \end{cases} \\
        A &= \{ q_0\}.
    \end{align*}

    Similarly, the automaton for addition in lsd-format is $M=(Q, \Sigma_k, \delta, q_0, A)$ where 
    \begin{align*}
        Q &= \{q_0, q_1\}, \\
        \delta(q, [a, b, c]) &=
                    \begin{cases}
                q_0, &  \text{if $q=q_0, c=a+b$}; \\
                q_0, &  \text{if $q=q_1, c=a+b+1$}; \\
                q_1, &  \text{if $q=q_1, c=a+b-k+1$}; \\
                q_1, &  \text{if $q=q_0, c=a+b-k$}; \\
            \end{cases} \\
        A &= \{ q_0\}.
    \end{align*}
Correctness is left to the reader.
\end{proof}

Now we consider what happens when we specify one summand to be a fixed constant.
\begin{theorem} \label{theorem:[x]+c=[y]-construction}
Let $c \geq 1$ be a fixed integer constant.
For both lsd- and msd-first input, there exists an automaton with $O(\log_k c)$
states accepting $x \times y$ where $[x]_k + c = [y]_k$.
\end{theorem}

\begin{proof}
    We first design the automaton $M_c$ accepting $[x]_k + c = [y]_k$ for msd-first input.
    The automaton for the special case of base $2$ can be created using the construction provided by Wolper and Boigelot \cite{wolper_construction_2000}, which is similar to what we do here. 
    
    On input $x \times y$, we need to keep track of the difference $D_k(x, y) = [y]_k - [x]_k$ and accept the input if and only if the difference is $c$. Each state of the automaton is a possible difference $d$, but we need to find a suitable range $I$ for $d$ so that the automaton is actually finite.
    We have
  \begin{equation}
    [yb]_k - [xa]_k = k([y]_k-[x]_k) + b - a.
    \label{equation:[yb]-[xa]}
    \end{equation}
    Therefore, if $D(x, y)$ is outside the range $[0, c]$ it will not return to it for any right extension of $x\times y$, and  we can take $I$ to be the interval $[0, c]$.
    The transition function is based on Eq.~\eqref{equation:[yb]-[xa]}.  After reading each input letter pair, the difference either stays the same or increases (unless it is no longer in the range $I$).
    Furthermore, if we consider a state $d \neq 0$, there are only two other states that have a direct transition to $d$, namely $\lfloor d/k \rfloor$ and $\lceil 
d/k \rceil$.

So we construct the automaton $M_c = (Q, \Sigma^2_k, \delta, q_0, A)$ where $\delta(d, [a, b]) = kd + b - a$, $q_0 = 0$, and $A=\{c\}$. To create the set $Q$, we start with $\{ c\}$. Then we add the states  $c' = \lfloor c/k \rfloor$ and $c'' = \lceil 
c/k \rceil$ to $Q$. We know $|c'' - c'|\leq 1$; therefore, the set $\{ \lfloor c'/k \rfloor,  \lfloor c''/k \rfloor, \lceil c'/k \rceil, \lceil c''/k \rceil\}$ has size at most $2$.
We repeat the same step (adding the floor and ceiling of division by $k$ to $Q$) until reaching the state $0$.
This step is repeated $O(\log_k c)$ times and there are always at most two new states created at each step. Therefore, the number of states in $M_c$ is $O(\log _ k c)$.

We now design the automaton $M_c$ for lsd-first input. The automaton for the special case of base $2$ can also be created using the construction provided by Boudet and Comon \cite{boudet_diophantine_1996}.

For the lsd-first input, we only need to use the simple manual addition method. We first write $(c)_k = \cdots c_1 c_0$ with $\lfloor \log_k c +1\rfloor$ letters. At first, after reading each letter pair, the automaton keeps track of the carry ($0$ or $1$) and how many letters from $(c)_k$ have been processed (the level starting from $0$). After processing all letters in $(c)_k$, in an extra final level (level $\lfloor \log_k c + 1 \rfloor$), the automaton only keeps track of the carry.

So we construct the automaton $M_c=(Q, \Sigma^2_k, \delta, q_0, A)$ as follows. The set $Q$ consists of pairs $[l, r]$ where $l$ is the level and $r$ is the carry.
\begin{align*}
Q &= \{[l, r] \in \mathbb{N}^2: l \leq \lfloor \log_k c + 1 \rfloor, r \in \{0, 1\}\}, \\
\delta([l, r], [a, b]) &=
\begin{cases}
    [l, (r+a - b)/k], &  \text{if $l= \lfloor \log_k c + 1 \rfloor $}; \\
    [l+1, (r + c_l + a - b)/k], & \text{otherwise;}
\end{cases}
 \\
q_0 &= [0, 0], \\
A &= \{ [l, r] \in Q \suchthat l= \lfloor \log_k c + 1 \rfloor, r=0\}.
\end{align*}
Note that in the transition function, if the new carry $r$ computed is not $0$ or $1$, the automaton transitions to the dead state.

The automaton has $O(\log_k c)$ levels of states and each level has two states for carry $0$ and carry $1$; therefore, the automaton has $O(\log _k c)$ states.
\end{proof}

\begin{theorem}
    Let $c \geq 1$ be a fixed constant.
    For both lsd- and msd-first base-$k$ input, there exists an automaton with  $O(\log_k c)$
    states accepting $x \times y$ where $[x]_k - c = [y]_k$.
\end{theorem}

\begin{proof}
    For the msd-first base $2$ input, the construction provided by Wolper and Boigelot \cite{wolper_construction_2000} can be used here too.
    For the msd-first base $k$ input, recall the automaton $M_c$ for msd-first in the proof of Theorem \ref{theorem:[x]+c=[y]-construction}. We only need to modify the previous automaton. Here, instead of $[y]_k-[x]_k$, we keep track of $[x]_k - [y]_k$ and the transition function is
    \[
    \delta(d, [a, b]) = kd + a - b.
    \]
    The automaton construction, the range $I$, and acceptance criteria of $d=c$ stay the same.

    For the lsd-first base $2$ input, the construction provided by Boudet and Comon \cite{boudet_diophantine_1996} can be used here too.
    For the lsd-first base $k$ input, the lsd-first automaton from the proof of Theorem \ref{theorem:[x]+c=[y]-construction} can be modified. The new transition function is
    \[
    \delta([l, r], [a, b]) =
    \begin{cases}
    [l, (r+b - a)/k], &  \text{if $l= \lfloor \log_k c + 1 \rfloor$}; \\
    [l+1, (r+ c_l + b - a)/k], & \text{otherwise.}
    \end{cases}
    \]
    Other components of the automaton stay the same.
\end{proof}

Recall the definition of the $k$-adic valuation: $\nu_k (n) = e$ if $k^e \mid n$ but $k^{e+1} \nmid n$. 
Note that $\nu_k(c)$ is the number of consecutive $0$s at the end of $(c)_k$. 

We also use the following notation:  if some condition is described inside brackets, the brackets are Iverson brackets. The Iverson bracket $[P]$ equals $1$ if the condition $P$ is satisfied, and $0$ otherwise.  The next theorem appears to be new.
\begin{theorem} \label{theorem:+c-sc-formula}
Let $c \geq 1$ be a fixed constant.
For msd-first input, the number of states in the minimal DFA recognizing $x \times y$ where $[x]_k + c = [y]_k$ is exactly
\begin{equation}
2|(c)_k| + 1 - \nu_k(c) - [(c)_k \ \text{starts with } \ 1] - [k=2, (c)_2 \ \text{starts with} \ 10, \ \text{and} \ c \ \text{is not a power of} \ 2].
\label{eq22}
\end{equation}
\end{theorem}

\begin{proof}
    Consider the construction of the automaton $M_c$ for msd-first input described in the proof of Theorem \ref{theorem:[x]+c=[y]-construction}.
    Let us group the new states created at each step of the construction into a \textit{level} including an initial level for the state $c$. We call a level with $1$ state a single-state level and a level with $2$ states a double-state level.

    In its simplest form, 
    our construction creates an automaton with consecutive double-state levels and an initial single-state level for $c$. At each step of the construction, the double-state levels are created by computing both the ceiling and floor of an integer divided by $k$ or the floor and ceiling of two consecutive integers divided by $k$. The number of times we need to divide $c$ by $k$ and apply floor and ceiling to get to the $0$ state is $|(c)_k|$. Therefore, in the simple form, there are $2|(c)_k| + 1$ states.
    The automaton $M_7$ for base $3$ is an example of the simple form.

    However, for some values of $k$ and $c$, some double-state levels in the simple form collapse to single-state levels, causing the number of states to be lower than the simple form. We study these cases, how they are reflected in Eq.~\eqref{eq22}, and why they are the only possible such cases.

\bigskip\noindent{\bf
    Case 1:} A positive power of $k$ divides $c$.

    If $k$ divides some integer $c'$, then $\lfloor c'/k \rfloor = \lceil c'/k \rceil$.  As a result, in this case, instead of the last level being single-state, the last $\nu_k(c)+1$ levels are single-state. So Case 1 is covered by the $-\nu_k(c)$ in Eq.~\eqref{eq22}. The automaton $M_6$ in base $3$ is an example of this case.

\bigskip\noindent{\bf 
    Case 2:} \label{case:2} The word $(c)_k$ starts with $1$, or $k=2$ and $(c)_2$ starts with $10$.

    The shortest path from state $0$ to state $c$ indicated by the transitions is $0^{|(c)_k|} \times (c)_k$. If $(c)_k$ starts with $1$, the shortest path starts with $[0, 1]$. As a result, the state $1$ is included in the automaton, but the states $0$ and $1$ cannot be on the same level; otherwise, the shortest path to $c$ does not start with $[0, 1]$. The automaton $M_4$ for base $3$ is an example.

    Similarly, if $k=2$ and $(c)_k$ starts with $10$, the shortest path starts with $[0, 1]$ followed by $[0, 0]$. As a result, the states $0, 1, 2$ are included in the automaton and not only $0$ and $1$ cannot be on the same level, but also $1$ and $2$ cannot be on the same level. The automaton $M_5$ for base $2$ is an example.
    
    If $k=2$, the word $(c)_2$ starts with $10$, and also $c$ is a power of $2$, the state reduced in the second level is already taken into account in $-\nu_k(c)$ described in Case 1. Therefore, Case 2 is covered by the $-[(c)_k \ \text{starts with } \ 1] - [k=2, (c)_k \ \text{starts with} \ 10, \ \text{and} \ c \ \text{is not a power of} \ 2 ]$ in Eq.~\eqref{eq22}, while ensuring absence of overlap with Case 1.

    \bigskip

    It remains to show that besides Cases 1 and 2, no additional reductions in the number of states compared to the simple form is possible.

    Consider a double-state level consisting of states $d$ and $d+1$ following a single-state level $d'$. In this case, $d$ transitions to $d+1$ (proof by contradiction). We have $d+1 = kd + b -a $ for some $a, b \in [0, k-1]$, $d \neq 0$, and $k \geq 2$. The only possible solutions to this equation are the following:
\begin{itemize}[nosep]
    \item $d=1,b=0, a=k-2$,
    \item $d=1, b=1, a=k-1$,
    \item $d=2, b=0, a=2k-3, k=2,$
\end{itemize}    
which are covered in Case 2.

    If some positive integer $c'$ is not divisible by $k$, then the floor and ceiling of $c'/k$ are different numbers, except for the special case already discussed in Case 2. So Case 1 covers all the other single-state levels not counted in Case 2.

    Therefore, the formula introduced in the theorem correctly computes the number of states in the minimal automaton $M_c$.
\end{proof}

\subsection{Multiplication}

In addition to providing an msd-first automaton construction for recognizing the relation $[x]_k + [y]_k = [z]_k$, Waxweiler \cite[Proposition 4.4.3]{waxweiler_caractere_2009} provided one for recognizing $n[x]_k = [y]_k$. Later,
Charlier et al.~\cite{charlier_minimal_2021} provided a construction for recognizing $n[x]_k +c = [y]_k$. Furthermore, the automaton for $n[x]_k + c =[y]_k$ with msd-first base-$2$ input (resp., lsd-first base-$2$ input) can be created by the methods explained  by Wolper and Boigelot \cite{wolper_construction_2000} (resp., Boudet and Comon \cite{boudet_diophantine_1996}). The following theorem considers the automaton accepting $n[x]_k + c = [y]_k$ in the case $c < n$ with both lsd- and msd-first base $k$ input.

\begin{theorem} \label{theorem:n[x]+c=[y]-sc}
Let $n \geq 1$, $0 \leq c < n$ be fixed constants.
For both lsd- and msd-first input, there is an $n$-state automaton accepting $x \times y$ where $n[x]_k + c = [y]_k$.
\end{theorem}

\begin{proof}
    For the msd-first input, let $x, y $ be two words of the same length in $\Sigma_k^*$.
    We define $D_k (x,y ) = [y]_k - n [x]_k$.  The basic idea
    is to define a DFA $M_{n,c, k} = (Q,\Sigma, \delta, q_0, A)$
    such that $\delta(q_0, x \times y) = D_k(x,y)$, but some
    additional consideration is needed, since $D_k$ could
    be arbitrarily large in absolute value.  This would seem to
    require infinitely many states.
    
    However, it is easy to see that
    \begin{itemize}
    \item [(a)] If $D_k(x,y) \leq -1$, then $D_k(xa,yb) \leq -1$ for 
    $a, b \in \Sigma_k$.
    \item [(b)] If $D_k(x,y) \geq n$, then $D_k(xa,yb) \geq n$ for
    $a, b \in \Sigma_k$.
    \end{itemize}
    We can therefore take our set of states
    $Q$ to be the set of all possible
    differences; namely, $Q = \{0,1,\ldots, n-1 \}$.
    Thus, if an input to the automaton
    $x \times y$ ever causes $D_k(x,y)$ to fall outside the
    closed interval $[0,n-1]$, then no right extension of this
    input could satisfy the acceptance criteria.  So in this case we can go to the
    (unique) dead state.
    
    The transition rule for the automaton is now easy to
    compute.  Suppose $d = D_k (x,y)$.
    After reading $xa \times yb$, where $a, b \in
    \Sigma_k$ are single letters, the new difference is 
    $$[yb]_k- n [xa]_k = k[y]_k+b - n(k[x]_k + a) =
    kd+b-na.$$  
    This gives the transition rule of the automaton:
    $\delta(d, [a,b]) = kd+b-na$,
    provided $kd+b-na$ lies in the interval $[0,n-1]$; otherwise
    the automaton transitions to the dead state.    A state $d$ is
    accepting if and only if $d = c$.

    For the lsd-first input, each state corresponds to a carry in the range $[0, c]$. The initial state is $c$ to account for the $+c$. By reading $[a, b]$, the transition function computes the multiplication by $n$ and addition by the carry and the automaton transitions to the new carry state. Since $c < n$, a simple proof by induction shows that the maximum carry possible is $c \leq n-1$ and the range $[0, c]$ for states is sufficient.

    So for lsd-first format we construct the automaton $M=(Q, \Sigma_k^2, \delta, c, A)$ where 
    \begin{align*}
        Q& =\{r\in \mathbb{N} \suchthat r \in [0, c]\}, \\
        \delta(r, [a, b]) & = (na + r - b)/k, \\
        A &= \{r \in Q \suchthat r =0\}.
    \end{align*}
    This automaton has $c+1 \leq n$ states.

\end{proof}

\section{Linear Subsequences of Automatic Sequences} \label{section:linear-subsequences}

We now turn to examining the state complexity of transformations on automatic sequences.

Zantema and Bosma \cite[Theorem 16]{zantema_complexity_2022} showed that for lsd-first input, if $(h(i))_{i \geq 0}$ is a $k$-automatic sequence generated by a DFAO of $m$ states, then $(h(i+1))_{i \geq 0}$ is generated by a DFAO of at most $2m$ states. Furthermore, in the same theorem they showed that for msd-first input, the sequence $(h(i+1))_{i \geq 0}$ is generated by a DFAO of at most $m^2$ states and there is an example where this bound is tight. 
Additionally, they \cite[Theorem 17]{zantema_complexity_2022} showed that the same bounds hold for $(h(i-1))_{i \geq 0}$ where $h(0)$ is any chosen letter from the sequence alphabet.

However, they did not show that their upper bound of $2m$ states for 
$(h(i+1))_{i \geq 0}$ in the lsd-first case is tight.  We suspect it is not.  In the following result, we show there are examples attaining the very slightly weaker bound of $2m-1$ states.
\begin{theorem} \label{theorem:h(i+1)-sc-lowerbound}
    For lsd-first input, for each $m \geq 2$, 
    there exists an $m$-state DFAO $M$ generating the sequence $(h(i))_{i \geq 0}$ such that at least $2m-1$ states are required in the DFAO $M'$ generating $(h(i+1))_{i \geq 0}$. 
\end{theorem}

\begin{proof}
We use an adaptation of the Myhill-Nerode theorem \cite[\S 3.4]{hopcroft_introduction_1979} for DFAOs. The adaptation for DFAOs is also used by Zantema and Bosma \cite{zantema_complexity_2022} and the state equivalence for DFAOs is defined by van Spaendonck \cite{van_spaendonck_automatic_2020}.

Let $k=2$, $m \geq 3$, and define the lsd-first DFAO $M$ generating $(h(i))_{i \geq 0}$ as follows.
    \begin{align*}
        Q &= \{0, \ldots, m-1\}, \\
        \delta(q, a) &=
        \begin{cases}
            (q+1) \bmod m, &  \text{if $a=1$}; \\
            q+1, &  \text{if $a=0$ and $ q \in \{0, \ldots, m-3\}$}; \\
            q, & \text{if $a=0$ and $q \in \{m-2, m-1\}$;}
        \end{cases}
         \\
        q_0 &= 0, \\
        \Delta &= \Sigma_2, \\
        \tau(q) &=
        \begin{cases}
            1 , &  \text{if $q=m-1$}; \\
            0 , &  \text{otherwise}.
        \end{cases}
    \end{align*}

Let $M'$ be the lsd-first DFAO generating the sequence $(h(i+1))_{i \geq 0}$.  
    We show that there are  $2m-1$ words $w_i$ such that for all $w_i, w_j$ where $i \neq j$, there exists $w_{ij}$ such that the output of $M'$ on input $w_iw_{ij}$ is different from  the output of $M'$ on input $w_j w_{ij}$. In this case, we say $w_i$ and $w_j$ are {\it distinguishable.}

    The $2m-1$ pairwise distinguishable words are the following.
    \begin{itemize}[nosep]
        \item $0^i$ where $1 \leq i \leq m-2$,
        \item $1^i$ where $0 \leq i \leq m-2$,
        \item $0^{m-2}1$,
        \item $0^{m-2}11$.
    \end{itemize}
    
    We need to show each pair of words is distinguishable.  There are eight cases to consider.

    \medskip

    \noindent{\it Case 1:} $0^i$ and $0^j$ where $1 \leq i < j \leq m-2$.
    Concatenate $0^{m-j-2}1$.
    The new words are $0^{i+m-j -2 }1$ and $0^{m-2}1$.
    The output of the automaton $M'$ on these words is then the same as the output of the automaton $M$ on inputs $1 0^{i-j+m-3}1$ and $1 0^{m-3}1$, which are $0$ and $1$. Therefore, the two words are distinguishable.

    Proving the rest of the cases follows the same logic.

    \medskip

    \noindent{\it Case 2:} $1^i$ and $1^j$ where $0 \leq i < j \leq m-2$.
    Concatenate $1^{m-j-2}0$. 

    \medskip

    \noindent{\it Case 3:} $0^{m-2}1$ and $0^{m-2}11$.
    Concatenate $\lambda$.

    \medskip

    \noindent{\it Case 4:} $0^i$ and $1^j$ where $1 \leq i\leq m-2$ and $0\leq j \leq m-2$. 
    Concatenate $1^{2m-i}0$.

    \medskip

    \noindent{\it Case 5:} $0^i$ and $0^{m-2}$1 where $ 1 \leq i \leq m-2$.
    Concatenate $\lambda$.

    \medskip

    \noindent{\it Case 6:} $0^i$ and $0^{m-2}$11 where $1 \leq i \leq m-2$.
    Concatenate $1^{m-i-1}$.

    \medskip

    \noindent{\it Case 7:} $1^i$ and $0^{m-2}$1 where $0 \leq i \leq m-2$.
    Concatenate $1^{m-i-1}$.

    \medskip

    \noindent{\it Case 8:} $1^i$ and $0^{m-2}11$ where $0 \leq i \leq m-2$.
    Concatenate $1^{m-i-2}$.

    \medskip

    \noindent Therefore, the DFAO $M'$ for $(h(i+1))_{i \geq 0}$ requires at least $2m-1$ states.
\end{proof}

\begin{theorem} \label{theorem:h(i+j)}
Suppose $(h(i))_{i \geq 0}$ is an automatic sequence generated by a DFAO of $m$ states with msd-first input.
Then there is a DFAO of $O(m^2)$ states with msd-first input generating the
two-dimensional automatic sequence $(h(i+j))_{i, j \geq 0 }$.
\end{theorem}

\begin{proof}
    Let the DFAO generating the sequence $(h(i))_{i \geq 0}$ be $M=(Q, \Sigma_k, \delta, q_0, \Delta, \tau)$. We want to design a DFAO $M'$ for $(h(i+j))_{i, j \geq 0}$. The basic idea is to perform $i+j$ while reading the representation of $i$ and $j$ in parallel; however, upon reading a letter pair $[a, b]$ and before reading the next letter pair, we do not know what the letter corresponding to $[a, b]$ in $i+j$ is. For example, $[00]_3 + [01]_3 = [01]_3$ but $[01]_3 + [02]_3 = [10]_3$; the first letter pairs read are the same but the first letters in the addition results are not. So each state in the automaton $M'$ includes two states from $M$: one of them corresponds to the case where a carry is not added to $a+b$, and one to the case where a carry is added.

    More formally, we construct the automaton $M'=(Q', \Sigma_k^2, \delta', q'_0, \Delta, \tau')$ as follows.
    \begin{align*}
    Q' &= Q^2, \\
    \delta'([q, p], [a, b]) &=
    \begin{cases}
        [\delta(q, a+b), \delta (q, a+b+1)], &  \text{if $a+ b \leq k-2$}; \\
        [\delta(q, a+b), \delta(p, a+b+1-k)], &  \text{if $a+b = k-1$}; \\
        [\delta(p, a+b-k), \delta(p, a+b+1-k)], & \text{if $a+b \geq k$;}
    \end{cases}
     \\
    q'_0 &= [\delta(q_0, 0), \delta(q_0, 1)], \\
    \tau'([q, p]) &= \tau(q).
    \end{align*}

The automaton is constructed so that if $[q, p]$ is a state in the automaton reachable by $x \times y$, then $[q, p] = [h'([x]_k+[y]_k), h'([x]_k+[y]_k+1)]$. So the reachable states are length-$2$ subsequences of the interior sequence $(h'(i))_{i \geq 0}$. Therefore, based on Theorem \ref{theorem:subword-complexity}, there are $O(m^2)$ reachable states.
\end{proof}

We now turn to linear subsequences of automatic sequences.  More specifically, we consider the following problem.  Let $(h(i))_{i \geq 0}$ be a $k$-automatic sequence generated by a DFAO of $m$ states.  Let $n \geq 1$ and $c \geq 0$ be integer constants.  How many states are needed to generate 
$(h(ni+c))_{i \geq 0}$?

Zantema and Bosma \cite{zantema_complexity_2022} studied the specific case
where $n = k$ and $0 \leq c < k$.  They showed that for this case and for both msd-first and lsd-first input, the subsequence $(h(ki+c))_{i \geq 0}$ is generated by a DFAO of at most $m$ states.

In \cite[Proof of Theorem 6.8.1]{allouche_automatic_2003}, it is proved that if $(h(i))_{i \geq 0}$ is a $k$-automatic
sequence generated by an $m$-state DFAO with lsd-first input, then for $n \geq 1, c \geq0$ the linear subsequence $(h(ni+c))_{i \geq 0}$ is generated by an lsd-first DFAO with
$m(n+c)$ states.  (The theorem is phrased in terms of the $k$-kernel, but
the number of elements of the $k$-kernel of a sequence is the same as the number of states in a minimal lsd-first DFAO generating the sequence.)

In the special case that $0 \leq c < n$, this lsd-first bound was improved slightly to at most $mn$ states by
Zantema and Bosma \cite{zantema_complexity_2022}.
Furthermore, they provided examples where the same bound does not hold for msd-first input. They stated obtaining a bound in the msd-first case as an open question.  We address this in the following theorem, which provides a close and novel connection between state complexity and subword complexity.

\begin{theorem}\label{theorem:h(ni+c)-msd-sc}
Let $n \geq 1$, $c \geq 0$ be fixed integer constants.
Let ${\bf h} = (h(i))_{i \geq 0}$ be an automatic sequence generated by a DFAO of $m$ states with msd-first input, and let
${\bf h}' = (h'(i))_{i \geq 0}$ be the
interior sequence of $\bf h$.
\begin{itemize}
    \item[(a)] If $c < n$, then there is a DFAO of at most 
    $\rho_{{\bf h}'} (n)$ states
    generating $(h(ni+c))_{i \geq 0}$ with msd-first input.

    \item[(b)] If $c \geq n$, then there is a DFAO of at most
    $\rho_{{\bf h}'} (c+1)$ states
    generating $(h(ni+c))_{i \geq 0}$ with msd-first input.
    \end{itemize}
\end{theorem}

\begin{proof}
Let $(h(i))_{i \geq 0}$ be generated by a DFAO
$M=(Q, \Sigma_k, \delta, q_0, \Delta, \tau)$.
\begin{itemize}
    \item[(a)] Suppose $c < n$. We want to create a DFAO generating $(h(ni+c))_{i \geq 0}$. 
The idea is that we construct $M'$ such that on input $x$ it reaches a state given by the following subword of the interior sequence $(h'(i))_{i \geq 0}$:
    $$[h'(n[x]_k), h'(n[x]_k+1), \ldots, h'(n[x]_k+n-1)].$$
    Then, since $c \leq n-1$, we only need to obtain $h'(n[x]_k+c)$ by mapping the subword $[r_0, \ldots, r_{n-1}]$ to $r_c$.

    Let $S_n(h')$ be the set of length-$n$ subwords of the sequence $(h'(i))_{i \geq 0}$. We define $M'=(Q', \Sigma_k, \delta', q_0', \Delta, \tau')$ as follows:
    \begin{align*}
        Q' &= \{[r_0, \ldots, r_{n-1}] \in S_n(h')\}, \\
        \delta'([r_0, \ldots, r_{n-1}], a) &= [s_0, \ldots, s_{n-1}] \\ & \text{where} \ s_d = \delta(r_j,b), \ j = \lfloor (na+d)/k \rfloor, \ b = (na+d) \bmod k, \\
        q'_0 &= [h'(0), h'(1), \ldots, h'(n-1)],\\
        \tau'([r_0, \ldots, r_{n-1}]) &= \tau(r_c).
    \end{align*}
The states of $M'$ are the length-$n$ subwords of the interior sequence $(h'(i))_{i \geq 0}$.  Note that not all states are necessarily reachable from the initial state.

First, let us see that
\[
s_d = \delta(r_j,b), \quad j = \lfloor (na+d)/k \rfloor, \quad b = (na+d) \bmod k
\]
in the transition function is well-defined.   We have
$0 \leq a \leq k-1$, and the definition of $b$
implies that $0 \leq b < k$.
We need to verify that $0 \leq j \leq n-1$.  We do this as follows: 
$j = \lfloor (na+d)/k \rfloor$
implies that $j \leq (n(k-1)+d)/k$.
Since $d<n$, we see this implies that $kj < nk$.  Hence $j < n$,
as desired.

Now let us prove by induction on
input length $|x|$ that 
\begin{equation}
\delta'( q'_0, x) = 
[h'(n[x]_k), \ldots, h'(n[x]_k + n-1)].
\label{equation:h(ni+c)}
\end{equation}
The base case is $x = \lambda$. Here the result follows by the definition of $q'_0$.

Otherwise, assume Eq.~\eqref{equation:h(ni+c)} is true for $|x|$; we prove it for $|x|+1$.  Write
$y = xa$.  So $[y]_k = k[x]_k + a$.  Then
$\delta'( q'_0, y) =
\delta'(\delta'(q'_0, x),a)$.
By induction
$\delta'( q'_0, x) =  [r_0, r_1,\ldots r_{n-1}]$ where $r_j = h'(n[x]_k + j)$.  By the definition of $\delta'$ we have
$\delta'(q'_0, y) = 
[s_0, s_1 \ldots, s_{n-1}]$ where
\begin{align*}
    s_d &= \delta(r_j, b) = \delta(h'(n[x]_k +j), b) 
    =\delta(q_0, k(n[x]_k+j)+b)\\
&= \delta(q_0, kn[x]_k + na+d)
= h'(kn[x]_k + na+d)
= h'(n(k[x]_k+a) + d)\\
&= h'(n [y]_k + d),
\end{align*}
as desired.

\item[(b)]  Suppose $c \geq n$.  The proof proceeds exactly as in (a), except now the states are the length-$(c+1)$ (instead of length-$n$) subwords of
$(h'(i))_{i \geq 0}$. The idea is that on input $x$, the DFAO $M'$
    reaches a state given by the
    subword 
    \begin{equation}
        [h'(n[x]_k), h'(n[x]_k+1), \ldots, h'(n[x]_k+c)].
    \end{equation}
Similarly, the states are length-$(c+1)$ subwords of the interior sequence $(h'(i))_{i \geq 0}$.

It is now not difficult to check, by exactly the same sort of calculation as in (a), that this is enough to construct the necessary subword giving the desired image, which will always be the last component of the state.
\end{itemize}
\end{proof}

\begin{corollary} \label{corollary:h(ni+c)-msd-sc}
Let $n \geq 1$, $c \geq 0$ be fixed integer constants.
Let $(h(i))_{i \geq 0}$ be an automatic sequence generated by a DFAO of $m$ states with msd-first input.   
\begin{itemize}
    \item[(a)] If $c < n$, then there is a DFAO of at most $k n m^2$ states
    generating the automatic sequence $(h(ni+c))_{i \geq 0}$ with msd-first input.
    \item[(b)] If $c \geq n$, then there is a DFAO of at most $k (c+1) m^2$ states generating the automatic sequence $(h(ni+c))_{i \geq 0}$ with msd-first input.
\end{itemize}
\end{corollary}

\begin{proof}
\leavevmode
\begin{itemize}
    \item[(a)]
    From Theorem \ref{theorem:subword-complexity}, we know $\rho_{{\bf h}'} (n) \leq knm^2$.
    \item[(b)]
    Similarly, we know $ \rho_{{\bf h}'} (c+1) \leq k(c+1)m^2$.
    \end{itemize}
\end{proof}

We do not know if the bounds in
Corollary~\ref{corollary:h(ni+c)-msd-sc} are tight.  This seems worthy of further study, so we state it as an open problem.
\begin{openproblem}
Find a class of automatic sequences
$(h_m(i))_{i \geq 0}$, generated by a DFAO of $m$ states with msd-first input, such that there is a constant $C$ so the linear subsequences
$(h_m(ni))_{i \geq 0}$ have at least the state complexity $C n m^2$ for
infinitely many $n$.
\end{openproblem}

\subsection{Connection with Subword Complexity}

We now consider an interesting application of this connection between state complexity and subword complexity.  Given an automatic sequence ${\bf h} = (h(i))_{i \geq 0}$ we would like to obtain an upper bound on the subword complexity of its linear subsequence $(h(ni))_{i \geq 0}$. 
\begin{theorem} \label{theorem:rho-h(ni)}
   Let ${\bf h} = (h(i))_{i \geq 0}$ be a $k$-automatic sequence and
   ${\bf h}'$ its interior sequence.  Let
   ${\bf h}_n$ denote the
   linear subsequence
   $(h(ni))_{i \geq 0}$ and
   similarly for
   ${\bf h}'_n$.
   Then $\rho_{{\bf h}_n} (\ell) \leq k \ell \rho_{{\bf h}'} (2n)$ for all $l, n \geq 1$.
\end{theorem}
\begin{proof}
    Clearly $\rho_{{\bf h}_n}(\ell) \leq \rho_{{\bf h}'_n}(\ell)$.
    Let $M=(Q, \Sigma_k, \delta, q_0, \Delta, \tau)$ be an automaton generating ${\bf h}'_n$ built by the construction in Theorem \ref{theorem:h(ni+c)-msd-sc}, and let $\hat{\bf M}$ be the sequence generated by the intrinsic DFAO  $\hat{M}=(Q, \Sigma_k, \delta, q_0, Q, \iota)$ where $\iota:Q \rightarrow  Q$ is the identity map.
    By Theorem~\ref{theorem:subword-complexity}, we have $\rho_{\bf h_n'}(\ell) \leq k \ell\rho_{\hat{\bf M}}(2)$. However, each length-$2$ subword of $\hat{\bf M}$ corresponds, from the construction of  Theorem~\ref{theorem:h(ni+c)-msd-sc}, to a distinct length-$2n$ subword of ${\bf h}'$.  Thus $\rho_{\hat{\bf M}} (2) \leq \rho_{{\bf h}'} (2n)$.  Putting this all together gives
    $$\rho_{{\bf h}_n} (\ell) \leq \rho_{{\bf h}'_n} (\ell) \leq k \ell \rho_{\hat{\bf M}}(2) \leq k \ell \rho_{{\bf h}'} (2n),$$
as desired.
 \end{proof}

\begin{remark}
  A superficially similar notion of complexity was studied by Konieczny and M\"ullner \cite{konieczny_arithmetical_2024}. However, the particular notion of complexity they studied is the number of length-$\ell$ subwords of all of the subsequences $(h(ni+c))_{i \geq 0}$ for integers $n\geq 1$ and $c \geq0$; this notion was first introduced by Avgustinovich et al.~\cite{avgustinovich_arithmetical_2003}.  In contrast, the complexity we discuss in this paper is the number of length-$\ell$ subwords of a single subsequence $(h(ni))_{i \geq0}$ with fixed $n$.
\end{remark}

\subsection{Thue-Morse Sequence}
We now apply some of the previous theorems to the famous Thue-Morse sequence ${\bf t} =(t(i))_{i \geq 0}$. The Thue-Morse sequence is
defined as the fixed point of the 
morphism $0 \rightarrow 01$,
$1 \rightarrow 10$ \cite{allouche_ubiquitous_1999}.

First we apply Theorem \ref{theorem:rho-h(ni)} to the Thue-Morse sequence.

\begin{corollary}
    Let ${\bf t}_n$ denote the linear subsequence $(t(ni))_{i \geq 0}$. We have $\rho_{{\bf t}_n}(\ell) \leq  \frac{40}{3}\ell n$ for $\ell, n \geq 1$.
\end{corollary}
\begin{proof}
    Results from Avgustinovich \cite{avgustinovich_number_1994} imply that $\rho_{\bf t}(n) \leq \frac{10}{3} n$. So by Theorem \ref{theorem:rho-h(ni)},
we have $\rho_{{\bf t} _n}(\ell) \leq 2 \cdot \ell \cdot \frac{10}{3} \cdot 2n = \frac{40}{3}\ell n$.
\end{proof}

Now we apply Theorem \ref{theorem:h(ni+c)-msd-sc} to the Thue-Morse sequence $(t(i))_{i \geq 0}$ and $c=0$. 

\begin{theorem} \label{theorem:t-sc}
    For all integers $n \geq 1$, the number of states in the minimal DFAO generating $(t(ni))_{i \geq 0}$ with msd-first input is $\rho_{\bf t}(n/\nu_2(n))$.
\end{theorem}

In order to prove Theorem \ref{theorem:t-sc}, we first need Lemma \ref{lemma:t(ni)-reachability} and Theorem \ref{theorem:t-construction-optimal}.

\begin{lemma} \label{lemma:t(ni)-reachability}
    For an odd integer $n$, all length-$n$ subwords of the Thue-Morse sequence $\bf t$ are reachable states in the automaton generating $(t(ni))_{i \geq 0}$ with msd-first input built by the construction provided in Theorem \ref{theorem:h(ni+c)-msd-sc}.
\end{lemma}

\begin{proof}
    We use results by Blanchet-Sadri et al.~\cite[Lemma~12]{blanchet-sadri_abelian_2014} by setting $i=0$ in their lemma. As a result, for each subword of the Thue-Morse sequence and odd integer $n$, there exists integer $i'$ such that the subword occurs at index $ni'$ of the Thue-Morse sequence; in other words, the state corresponding to the subword is reachable by input $(i')_2$.
\end{proof}

\begin{theorem} \label{theorem:t-construction-optimal}
    For all odd integers $n \geq 1$, the construction provided in Theorem \ref{theorem:h(ni+c)-msd-sc} is optimal for $(t(ni))_{i \geq 0}$.
\end{theorem}

\begin{proof}
    Consider the automaton created by the construction for $(t(ni))_{i \geq 0}$. On input $x$ the automaton reaches the state $[t(n[x]), \ldots, t(n[x]+n-1)] $
    where the output is $t(n[x])$.
    Furthermore, by Lemma \ref{lemma:t(ni)-reachability} we know all length-$n$ subwords of the Thue-Morse sequence are reachable states in the automaton.
    
    Now suppose we have two distinct states in this automaton reachable by $x$ and $y$:
    \begin{align*}
        [t(n[x]), \ldots, t(n[x]+n-1)] \neq
        [t(n[y]), \ldots, t(n[y]+n-1)].
    \end{align*}
    We need to show they are distinguishable.
    
    If $t(n[x]) \neq t(n[y])$, the distinguishability is trivial. If $t(n[x]) = t(n[y])$, there is $0 < j < n$ such that $t(n[x] + j ) \neq t(n[y] +j)$. We need to find $z$ such that $t(n[xz]) \neq t(n[yz])$; in other words, we want $t(n[x0^{|z|}] + n[z]) \neq t(n[y0^{|z|}] + n[z])$.
    
    For any word $z$, since $t(n[x]) = t(n[y])$, we have $t(n[x0^{|z|}]) = t(n[y0^{|z|}])$. So, if there is word $z$ such that $\left\lfloor n[z]/2^{|z|}\right\rfloor = j$, the number of ones in $(n[x0^{|z|}] + n[z])_2$ is the sum of the number of ones in $(n[x] + j)_2$ and $(n[z] \bmod 2^{|z|})_2$; a similar result holds for $y$ instead of $x$. Since $t(n[x]+j) \neq t(n[y] + j)$, we have $t(n[xz]) \neq t(n[yz])$.
    By iterating different $z$, we get $ 0 \leq \lfloor n[z]/2^{|z|} \rfloor < n $. Therefore, such $z$ always exists.
\end{proof}

\begin{proof}[Proof of Theorem \ref{theorem:t-sc}]
    We know $t(2i) = t(i)$ and the automaton generating $(t(2ni))_{i \geq 0}$ is the same as the automaton generating $(t(ni))_{i \geq 0}$.
    Furthermore, Theorem \ref{theorem:t-construction-optimal} implies that for all odd $n$, the number of states in the minimal automaton generating $(t(ni))_{i \geq 0}$ is $\rho_{\bf t}(n)$.
    
    Therefore, in order to obtain the number of states in the minimal automaton generating $(t(ni))_{i \geq 0}$, we remove powers of $2$ from $n$ to get an odd integer $n' = n / \nu_2(n)$. The automaton generating $(t(n'i))_{i \geq 0}$ is the same as the automaton generating $(t(ni))_{i \geq 0}$ and the number of states in the minimal automaton generating $(t(n'i))_{i \geq 0}$ is $\rho_{\bf t}(n')=\rho_{\bf t}(n/\nu_2(n))$.
\end{proof}

Let $\statecomplexity(h(ni))$ denote the number of states in a minimal automaton generating the sequence $(h(ni))_{i \geq 0}$.

\begin{corollary} \label{corollary:t(ni)-sc}
    For all integers $r, p \geq 1$, we have
    \begin{align*}
        \begin{cases}
            \statecomplexity(t(2ri)) = \statecomplexity(t(ri)), \\
            \statecomplexity(t((2^p + r)i)) = 3 \cdot 2^p + 4(r-1), & \text{if $1 < r < 2^{p-1}$ and $r \bmod 2 = 1$};
            \\
            \statecomplexity(t((2^p + 2^{p-1} + r)i)) = 5\cdot 2^p + 2(r-1), & \text{if  $1 < r < 2^{p-1}$ and $r \bmod 2 = 1$}.
        \end{cases}
    \end{align*}
\end{corollary}

\begin{proof}
    We know $t(2i) = t(i)$ and the automaton generating $(t(2ri))_{i \geq 0}$ is the same as the automaton generating $(t(ri))_{i \geq 0}$. Therefore $\statecomplexity(t(2ri)) = \statecomplexity(t(ri))$.

    The subword complexity of the Thue-Morse sequence was studied by Brlek \cite{brlek_enumeration_1989}, de Luca and Varricchio \cite{de_luca_combinatorial_1989}, and Avgustinovich \cite{avgustinovich_number_1994}. Brlek \cite{brlek_enumeration_1989} proved that for $n \geq 3$ where $n=2^p + r'+1$ and $0 < r'\leq 2^p$ we have
    \[
    \rho_{\bf t}(n) = 
    \begin{cases}
        6\cdot 2^{p-1} + 4r', & \text{if $0 < r' \leq 2^{p-1}$}; \\
        8\cdot 2^{p-1} + 2r', & \text{if $2^{p-1} < r' \leq 2^p$}.
    \end{cases}
    \]

    By Theorem \ref{theorem:t-sc}, from the $ 0 < r'\leq 2^{p-1}$ case we get $\statecomplexity(t((2^p + r)i)) = 3 \cdot 2^p + 4(r-1)$. From the $2^{p-1}< r' \leq 2^p$ case we get $\statecomplexity(t((2^p + 2^{p-1} + r)i)) = 5\cdot 2^p + 2(r-1)$.
\end{proof}

\begin{remark}
At first glance it might appear that
Charlier et al.~\cite{charlier_state_2019, charlier_minimal_2021} already proved results similar to those in Corollary~\ref{corollary:t(ni)-sc}.  However, this is not the case. The DFAO we considered takes $(i)_2$ as input and outputs $t(ni)$. The DFA they considered takes $x$ as input and decides whether $[x]_k=ni+r$ where $t(i)=0$, $k$ is a power of $2$, and $0 \leq r<n$.
\end{remark}

Next we apply Theorem \ref{theorem:h(ni+c)-msd-sc} to the Thue-Morse sequence $(t(i))_{i \geq 0}$ and $n=1$. 

\begin{theorem}\label{theorem:t-shift}
    For all integers $c \geq 1$, the number of states in the minimal automaton generating $(t(i+c))_{i \geq 0}$ with msd-first input is at most $\frac{10}{3}c$.
\end{theorem}

\begin{proof}
    Consider the automaton generating $(t(i+c))_{i \geq 0}$ built by the construction in Theorem \ref{theorem:h(ni+c)-msd-sc} . On input $x$ the automaton reaches the state $[t([x]_k), \ldots, t([x]_k + c)]$ where the output is $t([x]+c)$. So the number of states in the minimal automaton generating $(t(i+c))_{i \geq 0}$ is at most $\rho_{\bf t}(c+1)$.
    
    We use the formula by Brlek \cite{brlek_enumeration_1989} previously mentioned in the proof of Corollary \ref{corollary:t(ni)-sc}.
    Let $c=2^p+r'$.
    If $0 < r' \leq 2^{p-1}$, then $r'\leq\frac{1}{3}c$ and we have $\rho_{\bf t}(c+1) = 6\cdot 2^{p-1} + 4r' = 3c + r' \leq \frac{10}{3}c$.
    If $2^{p-1} < r' \leq 2^p$, then $r'>\frac{1}{3}c$ and we have $\rho_{\bf t}(c+1)= 8\cdot 2^{p-1}+2r'=4c -2r' < \frac{10}{3}c$.
\end{proof}

Unlike the case of Theorem~\ref{theorem:t-construction-optimal}, where the construction in Theorem \ref{theorem:h(ni+c)-msd-sc} gives the minimal
automaton, in the case of Theorem~\ref{theorem:t-shift}, the construction in Theorem \ref{theorem:h(ni+c)-msd-sc} does not
necessarily give the minimal automaton, so the upper bound in Theorem~\ref{theorem:t-shift} is not tight.  The next theorem gives a weak lower bound on 
the state complexity of $(t(i+c))_{i\geq 0}$.

\begin{theorem}
The minimal DFAO computing $(t(i+c))_{i\geq 0}$ with msd-first input has more than
$c^{0.694}$ states for infinitely many $c$.
\end{theorem}

\begin{proof}
The idea is to consider the case $c = 2^n$ for $n \geq 1$.  In this case we 
show, using the adaptation of the Myhill-Nerode theorem for DFAOs,
that the minimal DFAO for $(t(i+c))_{i\geq 0}$ has at least $F_{n+1}$ states.
We let $S_n$ be the set of all binary words of length $n$
that start with $0$ and have no two consecutive $1$'s.  As is well-known,
this set has cardinality $F_{n+1}$.  We now show the members of this
set are pairwise inequivalent under the adaptation of the Myhill-Nerode theorem.  

Let $x,y \in S_n$ with $x \not= y$. We need to find a word $z$ 
such that $t([xz]_2 + 2^n) \not= t([yz]_2 + 2^n)$.
If $t([x]_2+2^n) \not= t([y]_2+2^n)$, we can choose $z = \lambda$, the
empty word.

Otherwise, assume $t([x]_2+2^n)=t([y]_2+2^n)$.
Since $x, y$ are distinct, there must be
some position $p$ (numbered from left to right, starting with index $1$)
at which $x[p] \not= y[p]$ where $w[i]$ is the $i$-th letter of $w$.  Without loss of generality, assume $x[p]=0$
and $y[p] = 1$.  Choose $z = 0^p$.  
Let $x'$ denote the binary representation of $[xz]_2 + 2^n$
and $y'$ denote the binary representation of $[yz]_2 + 2^n$.
Letting $|u|_1$ denote the number of $1$'s in a word $u$, we
clearly have $|x'|_1 = |x|_1 +1$, since a new $1$
arises from the addition of $2^n$.
However, $|y'|_1 = |y|_1$, since adding $2^n$ causes the $1$ at
position $p$ to become
$0$ and causes the $0$ at position
$p-1$ to become $1$, while no other bits 
change (since $y$ has no two consecutive $1$'s).
It follows that $t([xz]_2 + 2^n) \not= t([yz]_2 + 2^n)$, as desired.
Hence all the elements of $S_n$ are pairwise inequivalent, and so
the minimal DFAO for $(t(i+c))_{i \geq 0}$ has $F_{n+1}$ states.

Finally, from the estimate $F_{n+1} \geq \frac{1}{2}\alpha^n$
(which is easily proved by induction), where
$\alpha = (1+\sqrt{5})/2 \doteq 1.61803$ is the golden ratio,
we see that if $c = 2^n$, then the minimal DFAO has
at least 
$$F_{n+1} \geq \frac{1}{2} \alpha^n = \frac{1}{2} c^{\frac{\log \alpha}{\log 2}} \geq c^{0.694} $$
states, for $n$ large enough.
(Note that $(\log \alpha)/(\log 2) \doteq 0.69424$.)
\end{proof}

The exact number of states in the minimal automaton generating $(t(i+c))_{i \geq 0}$ for small values
of $c$ is given as sequence \seqnum{A382296} of the OEIS \cite{oeis_2025}. We state the following as an open problem.

\begin{openproblem}
What is a good formula for the exact number of states in the minimal automaton generating $(t(i+c))_{i \geq 0}$ with msd-first input, as a function of $c$?
\end{openproblem}

\subsection{Sequences Whose Shifts Have Small State Complexity}

A simple counting argument shows that if an automatic sequence $(h(i))_{i \geq 0}$ is not ultimately periodic, then the state complexity of the shifted sequence $(h(i+c))_{i \geq 0}$ cannot be $o((\log c)/(\log \log c))$.  We do not know any aperiodic example, automatic or not, that achieves state complexity $O((\log c)/(\log \log c))$ for all $c$-shifts, but we can get almost that close.

Let the infinite word ${\bf p}_2 = (P_2(i))_{i \geq 0}$ be defined as follows.
\[
P_2(i) = 
\begin{cases}
    1, & \text{if $i$ is a power of $2$}; \\
    0, & \text{otherwise.}
\end{cases}
\]

\begin{theorem}
    The number of states in the minimal automaton generating $(P_2(i+c))_{i \geq 0}$ with base-$2$ msd-first input is $O(\log c)$.
\end{theorem}
\begin{proof}
    We create an automaton $M=(Q, \Sigma, \delta, q_0, \Delta, \tau)$ where $\Sigma = \Delta = \{0, 1\}$ with $O(\log c)$ states generating $(P_2(i+c))_{i \geq 0}$.
    In order to create the automaton, we need to consider two possibilities for $i$ where $P_2(i+c)=1$. Let $c'$ be $2^s-c$ where $2^s$ is the smallest power of $2$ greater than $c$. For now let us ignore the case where $c$ is a power of $2$.

    \medskip

    \noindent{\it Case 1:} $i+c=2^j$ where $|(i)_2|\leq|(c)_2|$. In this case, the automaton reads $(c')_2$ in a series of states.
    
\medskip

    \noindent{\it Case 2:} $i+c=2^j$ where $|(i)_2|> |(c)_2|$. In this case, the automaton reads one or more $1$s and then reads $(c')_2$ padded with leading zeros so the length is the same as $(c)_2$.
    
    \medskip

    Now we provide a construction for $M$. Let $\delta(q_0, 0)=q_0$. We first address Case 1. Let $(c')=c''_0\cdots c''_l$.
    For $p\in [0, l]$, let $q''_p$ be a state in the automaton corresponding to $c''_0\cdots c''_p$ read and for $p \in [0, l-1]$, let $\delta(q''_p, c''_{p+1}) = q''_{p+1}$ and $\delta(q_0, 1)=q''_0$. Furthermore, let $\tau(q''_l)=1$.

    Next, we address Case 2. Let $c'_0\cdots c'_{|(c)_2|-1}$ be $(c')_2$ padded with leading $0$s so the length is $|(c)_2|$.
    For $p \in [0, |(c)_2|-1]$, let $q'_p$ be a state  corresponding to $c'_0\cdots c'_p$ read and for $p \in [0, |(c)_2|-2]$, let $\delta(q'_p, c'_{p+1})=q'_{p+1}$. Let $q_1$ be the state corresponding to reading the initial $1$s in Case 2. So $\delta(q_1, 1)=q_1$ and $\delta(q_1, c'_0)=q'_0$. We know $c'_0 = 0$ if $c$ is not a power of $2$. Furthermore, let $\tau(q'_{|(c)_2|-1})=1$.
    
    Now we need to define transitions so the automaton is deterministic and accepts $1^jc'_0\cdots c'_{|(c)_2|-1}$ for all $j \geq 1$.
    If there are no $0$s in $c_0''\cdots c''_l$, let $p''=l$; otherwise, let $p''$ be the first index starting from $0$ where $c''_{p''+1} = 0$. We set $\delta(q''_{p''}, 1)=q_1$; this handles $j \geq p''+2$. For $p \in [0, p''-1]$, we set $\delta(q''_{p}, 0)=q'_0$; this handles $1 \leq j' \leq p''$.  Let $p'$ be the first index starting from $0$ where $c''_0\cdots c''_l$ differs from $1^{p''+1}c'_0\cdots c'_{|(c)_2|-1}$, which can potentially be $l+1$. We set $\delta(q''_{p'-1}, c'_{p'-p''-1}) = q'_{p'-p''-1}$; this handles $j= p''+1$.

    So if $c$ is not a power of $2$, the DFAO $M=(Q, \Sigma_2, \delta, q_0, \Delta, \tau)$ can be formally defined as follows.
\begin{align*}
    Q &= \{q_0, q_1, q''_0, \ldots, q_l'', q'_0, \ldots, q'_{|(c)_2-1|}  \},\\
    \delta(q, a) &=
    \begin{cases}
        q_0, & \text{if $q=q_0, a=0$}; \\
        q_0'', & \text{if $q=q_0, a=1$}; \\
        q_1, & \text{if $q=q_1, a=1$}; \\
        q'_0, & \text{if $q=q_1, a=c'_0$}; \\
        q_1, & \text{if $q=q''_{p''}, a=1$}; \\
        q'_{p'-p''-1}, & \text{if $q=q''_{p'-1}, a=c'_{p'-p''-1}$}; \\
        q'_0, & \text{if $q=q''_p, a = 0$ for  $p \in [0, p''-1]$}; \\
        q''_{p+1}, & \text{if $q=q''_p, a = c''_{p+1}$ for  $p \in [0, l-1]$}; \\
        q'_{p+1}, & \text{if $q=q'_p, a = c'_{p+1}$ for  $p \in [0, |(c)_2|-2]$};
    \end{cases} \\
        \Delta &= \{0, 1\} , \\
            \tau(q) &= 
    \begin{cases}
        1, & \text{if $q=q''_l$ or $q=q'_{|(c)_2|-1}$}; \\
        0, & \text{otherwise}.
    \end{cases}
\end{align*}
The correctness of $M$ is explained above. The number of states in $M$ is  $2+ (l+1) + (c)_2 =O(\log c)$.

If $c$ is a power of $2$, then $c=c'$ and it is easy to see that the same idea works with less complications because $|(c)_2|=|(c')_2|$.
We just need to modify the automaton defined above and set $q''_i = q'_i$ for $i \in [0, l]$, $q_1 = q''_0$, $\tau(q_0)=1$ and omit the transitions involving $p'$ and $p''$. The number of states will remain $O(\log c)$.
\end{proof}

\section{Construction by B\"uchi Arithmetic} \label{section:Buchi} 

We now re-examine the previous relations and operations in light of their implementation in an interpretation of B\"uchi arithmetic (as, for example, in
the system \texttt{Walnut} \cite{mousavi_automatic_2021}) to find the time required for constructing their automata.  

B\"uchi arithmetic is the logical theory of the natural numbers together with addition and  the function $V_k(n)$, the largest power of $k$ dividing $n$, for some fixed $k \geq 2$.  This theory is powerful enough to express finite automata \cite{bruyere_logic_1994}; in what follows we treat DFAO's computing automatic sequences as more or less a primitive object, without requiring an explicit expression in terms of
$\langle \Enn, +, V_k \rangle$.

Here the total number of transitions of an automaton $M$ as a measure for the computational complexity of constructing $M$. Some of the complexity bounds in this section include the addition of $O(\log i)$ or $O(i^j)$ components for some constant $i$ which are not significant compared to multiplication in $O(\log i)$ or $O(i^j)$.

In particular, in its current version, {\tt Walnut} does not use the algorithms for creating the automata we discussed in previous sections.  Instead, more general techniques, suitable for any addable numeration system, are used.  This means that the size of intermediate automata may differ from the minimal automata we constructed, and the running time can similarly be larger.

In our analyses in this section, we focus on msd-first input. We need to make sure the final automata always have the correct output regardless of the number of leading zeros read; this is accomplished as follows. After creating an automaton $M$ with initial state $q_0$ and transition function $\delta$, we change the initial state to the state $q_0'$ where $\delta(q_0, [0^j \times \cdots \times 0^j])=q'_0$ for some $j$ and $\delta(q'_0, [0 \times \cdots \times 0])=q'_0$. The rationale is that our automata should always have the same correct output regardless of the number of leading zeros read and by creating the automaton in this manner we can always assume we have considered enough leading zeros for the automaton computations to be correct. A similar idea for base-$2$ is mentioned by Wolper and Boigelot \cite{wolper_construction_2000}.

Let the msd-first $2$-state automaton for recognizing addition from Theorem \ref{theorem:addition-FA} be $M_{\add}$. We use $M_{\add}(x, y, z)$ to show the output of $M_{\add}$ on input $x \times y \times z$, which is $1$ if and only if $[x]_k + [y]_k = [z]_k$; we use a similar notation to refer to the output of other automata on a specific input. Recall that $M_{\add}$ states correspond to the difference $[z]_k - ([x]_k +[y]_k)$.

In {\tt Walnut} each DFA or DFAO created is minimized. If $n$ is the number of states in an automaton (and the number of transitions is $O(n)$), the minimization steps in the implementation can use an algorithm running in $O(n \log n)$ time  by Hopcroft \cite{hopcroft_n_1971} or Valmari \cite{valmari_fast_2012}. The algorithm by Hopcroft can be applied to DFAOs with slight modifications as described by van Spaendonck \cite{van_spaendonck_automatic_2020}.  
Here we sometimes only consider the effects of a partial minimization to facilitate analysis in the following theorems; however, the upper bound on state complexity and runtime complexity is correct regardless and the runtime complexity computed takes into account the computational cost of a full minimization.

We start with an analysis of how an automaton for a statement like $[x]_k = c$ is created, where $c$ is a fixed natural number constant.  Recall that in {\tt Walnut}, constants are entered as numbers written in decimal notation (base $10$).
There are well-known efficient algorithms for converting a number from base $10$ to base $k$ for an integer $k \geq 2$; see, for example, \cite[\S 4.4]{knuth_seminumerical_1997}, but these are not used by {\tt Walnut}. Instead, a more general technique, suitable for any numeration system, is used.

\begin{lemma} \label{lemma:=c-sc}
The minimal automaton for recognizing the relation
$[x]_k = c$ with msd-first input has $\lfloor \log_k c \rfloor + 2$ states. 
\end{lemma}

\begin{proof}
Such an automaton consists of a linear chain of nodes labeled with the base-$k$ representation of $c$, which has
length $\lfloor \log_k c \rfloor + 1$.  
\end{proof}

\begin{theorem} \label{theorem:=c}
    Let $c \geq 0$ be a fixed constant. The automaton recognizing the relation $[x]_k = c$ with msd-first input can be computed in $ O((\log^2  c) (\log \log c))$ time.
\end{theorem}

\begin{proof}
    We recursively create the automaton $M_{=c}$.

    The base case is $c=0$ and the automaton has $1$ state.
    
    If $c$ is even, we recursively compute a minimal DFA $M_{=(c/2)}$ recognizing $[y]_k = c/2$ on input $y$ and then use the first-order expression
    \[
    \exists y \ M_{=(c/2)}(y) \land M_{\add}(y, y, x)
    \]
    which is translated into an automaton.  We know the automaton for $M_{\add}$ has $2$ states and by Lemma \ref{lemma:=c-sc}, the automaton $M_{=(c/2)}$ has $\lfloor \log_k (c/2) \rfloor$ + 2 states. First we apply the product construction to $M_{=(c/2)}$ and $M_{\add}$ such that the $M_{=(c/2)}$ part processes $y$ and the $M_{\add}$ part processes $y \times y \times x$; the resulting automaton has $2 (\lfloor \log_k (c/2) \rfloor + 2)$ states. Each state in the product automaton is of the form $[q, p]$ where on input $x \times y$, $q$ is a state from the minimal automaton $M_{=(c/2)}$ representing the number of letters from $(c/2)_k$ matched with $y$ by Lemma \ref{lemma:=c-sc} and $p \in \{0, 1\}$ is a state from $M_{\add}$ equal to $[x]_k - 2[y]_k$. 
    Then we apply the $\exists y$ quantifier by removing the component corresponding to $y$ from the transitions and we get an NFA. In the NFA there are no two same transitions from a single state leading to different states. So after subset construction the resulting DFA is the same as the NFA. Then we do minimziation and the resulting DFA is the automaton from Lemma \ref{lemma:=c-sc} with $O(\log c)$ states.
    
    If $c$ is odd, we recursively compute a DFA $M_{=(c-1)}$ recognizing $[y]_k = c-1$ on input $y$ and then use the first-order expression
    \[
    \exists y \ M_{(c-1)}(y) \land M_{\add}(y, 1, x)
    \]
    which is translated into an automaton. Following a construction similar to the one described for the case where $c$ is even, we get the DFA from Lemma \ref{lemma:=c-sc} for constant $c$ with $O(\log c)$ states.

    Taking minimization steps into account, we get the following recursive formula for the time complexity:
    \[  T(c) =
\begin{cases}
     O((\log c) (\log \log c)) + T(c/2), & \text{if $c \bmod 2 = 0$}; \\
    O((\log c) (\log \log c))  + T(c-1), & \text{if $c \bmod 2 = 1$}.
\end{cases}
\]
So the total time is $ O((\log^2  c) (\log \log c))$.
\end{proof}

\begin{theorem} \label{theorem:[x]+c=[z]-time}
    Let $c \geq 0$ be a fixed constant. The automaton recognizing the relation $[x]_k + c =[z]_k$ with msd-first input can be computed in $O((\log^2 c) (\log\log c))$ time.
\end{theorem}

\begin{proof}
    Consider the $O(\log c)$ state automaton $M_{=c}$ from Theorem \ref{theorem:=c} recognizing $[y]_k = c$ on input $y$. We have $M_{=c}(y)=1$ if and only if $[y]_k=c$. 
    
    We can use the following expression to create the automaton $M_c$ recognizing $[x]_k + c=[z]_k$ on input $x \times z$.
    \[
    \exists y \ M_{\add}(x, y, z) \AND M_{=c}(y).
    \]
    First, we apply the $\land$ by product construction for $M_{\add}$ and $M_{=c}$ such that the $M_{\add}$ part processes $x\times y \times z$ and the $M_{=c}$ part processes $y$. The resulting automaton $M'_c$ has $O(\log c)$ states.
     Let $\pre_k(c, i)$ denote the first $i$ letters in $(c)_k$. On input $x\times y \times z$ where $y$ is $0^*y'$ and $y'$ is a prefix of $(c)_k$, the automaton $M'_c$ leads to a state $[q, p]$ where $q\in \{0, 1\}$ is from $M_{\add}$ and $p=|y'|$ is from $M_{=c}$ indicating how many letters from $y$ have been matched with $(c)_k$. Therefore, we have $q=[z]_k-([x]_k + [y]_k)$ and $[\pre_k(c, p)]_k = [y]_k$. So $[z]_k - [x]_k = q + [\pre_k(c, p)]_k$ and to each state $[q, p]$ reachable on $x \times y \times z$ in $M'_c$ a difference $d = [z]_k - [x]_k$ can be attributed. 
    
    Then we apply the $\exists y$ quantifier by removing the component corresponding to $y$ from the transitions and we get an NFA. In the NFA there are no two same transitions from a single state leading to different states. So after subset construction the resulting DFA is the same as the NFA. Therefore, there are $O(\log c)$ states in the DFA and to each state of the DFA reachable on $x\times z$ a difference $d=[z]_k-[x]_k$ can be attributed.

    From the proof of Theorem \ref{theorem:[x]+c=[y]-construction}, we know the automaton recognizing $[x]_k + c= [z]_k$ on input $x \times z$ only needs to keep track of a difference $[z]_k - [x]_k$ in range $[0, c]$. So after some minimization, the resulting automaton $M_c=(Q, \Sigma^2_k, \delta, q_0, A)$ can be formally defined as follows:
    \begin{align*}
        Q &= \{d \in \mathbb{N} \suchthat d \in [0, c]\}, \\
        \delta(d', [a, e]) &= kd' + e -a, \\
        q_0 &= 0, \\
        A &= \{c\}.
    \end{align*}
    Note that not all $d \in [0, c]$ are present in the minimized automaton.
    
    Considering the time spent on creating $M_{=c}$ and the minimization steps, it takes $O((\log^2 c) (\log\log c))$ time to create the automaton $M_c$ recognizing $[x]_k + c = [z]_k$ on input $x \times z$. 
\end{proof}

We continue by constructing the automaton $M_n$ recognizing $n[x]_k = [z]_k$ on input $x\times z$ where $n$ is a fixed constant. We have $M_n(x, z)=1$ if and only if $n[x]_k = [z]_k$.
There are several possible recursions for creating this automaton; for example, we could use the following:
\begin{itemize} [nosep]
    \item $\exists y \ M_{n-1}(x,y) \AND M_{\add}(x, y, z)$.
    \item  $\exists y_1, y_2 \ M_{\lfloor n/2 \rfloor}(x,y_1) \AND M_{\lceil n/2 \rceil}(x,y_2) \AND M_{\add}(y_1, y_2, z)$.
\end{itemize}
However, there is a third recursive implementation that is more efficient than the two above, and we describe it now.
The base case is $n = 1$.  Otherwise, either $n$ is even or $n$ is odd.

If $n$ is even, we
recursively compute a DFA $M_{n/2}$ recognizing $[y]_k =  (n/2)[x]_k$ on input $x \times y$
and then use the first-order expression
$$ \exists y \ M_{n/2}(x,y) \AND M_{\add}(y, y, z),$$
which is translated into an automaton by a direct product construction
for the $\AND$, and projection of the second coordinate to collapse
transitions on $y$. 

If $n$ is odd, we recursively compute
a DFA $M_{n-1}(x,y)$ recognizing $[y]_k = (n-1)[x]_k$ and use the expression
$$ \exists y \ M_{n-1}(x,y) \AND M_{\add}(x, y, z),$$
which is similarly translated into an automaton.

In both cases, the translation can potentially generate a nondeterministic automaton (by the $\exists y$ projection).  In this case, the subset construction is used to determinize the resulting automaton. Thus, at least in principle, this particular translation mechanism could take exponential time.
However,  we show that in this case, the subset construction does not cause an exponential blowup in the number of states.

\begin{theorem} \label{theorem:y=n[x]-time}
    Let $n \geq 1$ be a fixed constant.
    The automaton for recognizing $[y]_k = n[x]_k$ on msd-first input can be computed in $O(n \log^2 n)$ time. 
\end{theorem}

\begin{proof}
First by induction on $n$ we prove that the we can recursively create DFA $M_n$ with msd-first input recognizing $[y]_k = n[x]_k$ with at most $n$ states.

Base case: $n=1$.

In this case, a simple $1$-state automaton recognizes $[y]_k=[x]_k$. The single state corresponds to $[y]_k - [x]_k = 0$.

For the induction step, assume the result is true for all $n'<n$; we prove it for $n$.
There are two cases.

\bigskip\noindent{\it
Case 1:} $n$ is even.

By the induction hypothesis, we have a DFA $M_{n/2}$ with $n/2$ states corresponding to differences $[y]_k - (n/2)[x]_k$ in the range $[0, (n/2)-1]$.
In this case, we use the formula
$$ \exists y \ M_{n/2}(x,y) \AND M_{\add}(y, y, z).$$
First, we create an automaton $M'_{n}$ that is the product construction of $M_{n/2}$ and $M_{\add}$. So this automaton has $2(n/2) = n$ states. Consider an input $x \times y \times z$ for $M'_{n}$ leading to some state $[q, p]$ where $q$ is from $M_{n/2}$ and $p$ is from $M_{\add}$. By the construction of $M'_{n}$, we know $q=[y]_k - (n/2)[x]_k$ and $p=[z]_k - ([y]_k + [y]_k)$. So we have $2q+p = [z]_k - n[x]_k$. Since $p \in \{0, 1\}$ and $q \in \{0, \ldots, n/2 - 1\}$, all differences in $\{0, \ldots, n-1\}$ are present in $M'_{n}$ and they are only represented by a single state.
Therefore, the automaton $M'_{n}$ keeps track of the difference $[z]_k - n[x]_k$ on input $x \times y \times z$. The difference $[z]_k - n[x]_k$ that the automaton keeps track of is independent of $y$.

Then we apply the $\exists y$ quantifier by removing the component corresponding to $y$ from the transitions and we get an NFA. In the NFA there are no two same transitions from a single state leading to different states. So after subset construction the resulting DFA is the same as the NFA and we get the $n$-state automaton $M_{n}=(Q, \Sigma^2_k, \delta, q_0, A)$ as follows.
\begin{align*}
Q &= \{ d \in \mathbb{N}  \suchthat d \in [0, n-1] \}, \\
\delta(d, [a, e]) &= kd + e - n
a, \\
q_0 &= 0, \\
A &= \{q_0\}.
\end{align*}
Projecting away the $y$ component from transitions could have caused nondeterminism, but this is not the case.

\bigskip\noindent{\it
Case 2:} $n$ is odd.

By the induction hypothesis, we have a DFA $M_{n-1}$ with $n-1$ states corresponding to $[y]_k - (n-1)[x]_k$ in the range $[0, n-2]$.
In this case, we use the formula
$$ \exists y \ M_{n-1}(x,y) \AND M_{\add}(x, y, z).$$
First, we create an automaton $M'_{n}$ from the usual product construction for $M_{n-1}$ and $M_{\add}$. So this automaton has $2(n-1) $ states. Consider an input $x \times y \times z$ for $M'_{n}$ leading to some state $[q, p]$ where $q$ is from $M_{n-1}$ and $p$ is from $M_{\add}$. By the construction of $M'_{n}$, we know $q=[y]_k - (n-1)[x]_k$ and $p=[z]_k - ([x]_k + [y]_k)$. So we have $q+p = [z]_k - n[x]_k$. Therefore, to each state in the automaton $M'_{n}$ reachable on $x \times y \times z$, a difference $d=[z]_k - n[x]_k$ can be attributed. Furthermore, since $p \in \{0, 1\}$ and $q \in \{0, \ldots, n-2\}$, all differences in $\{0, \ldots, n-1\}$ are present in $M'_{n}$. Note that the difference $0$ corresponds to exactly one pair---namely $[0,0]$---and the difference $n-1$ corresponds to only one pair---namely $[n-2,1]$.
All the other differences $d$ (those in $\{ 1,2,\ldots, n-2\}$)  are represented by two states---namely, $[d-1,1]$ and $[d,0]$.

Next, we project away the $y$ component.  The resulting automaton is an NFA, and we use the subset construction to create a DFA.  Consider a state from $M_{n}$ reachable on $x \times z$. This state consists of some pairs $[q, p]$ from the NFA such that $[z]_k - n[x]_k = q+p$. As mentioned above, if $[z]_k-n[x]_k$ is $0$ or $n-1$, then there is only one such pair; otherwise, there are two such pairs.
 So the DFA has $O(n)$ states and the determinization did not cause a blow up in the number of states.

From the proof of Theorem \ref{theorem:n[x]+c=[y]-sc}, we know the automaton recognizing $n[x]_k= [z]_k$ on input $x \times z$ only needs to keep track of a difference $[z]_k - n[x]_k$ in range $[0, n-1]$. So after some minimization, the DFA  $M_{n}=(Q, \Sigma^2_k, \delta, q_0, A)$ can be formally defined as follows.
\begin{align*}
Q &= \{ d \in \mathbb{N}  \suchthat d \in [0, n-1] \}, \\
\delta(d, [a, e]) &= kd + e - na,
 \\
q_0 &= 0, \\
A &= \{q_0\}.
\end{align*}

Let $T(n)$ be the time used by the algorithm above to create $M_n$. We now analyze $T(n)$ based on the information above.

In Case 1, the number of transitions traversed for creating $M_n$ is as follows: $O(1)$ transitions for $M_{\add}$, $O(n)$ transitions for the product construction of $M_{n/2}$ and $M_{\add}$, $O(n)$ transitions for projecting away the $y$ component, $O(n)$ transitions for the determinized automaton, giving a total of $O(n)$ transitions in addition to the number of transitions traversed to create $M_{n/2}$.

In Case 2, the number of transitions traversed for creating $M_n$ is as follows: $O(1)$ transitions for $M_{\add}$, $O(n)$ transitions for the product construction of $M_{n-1}$ and $M_{\add}$, $O(n)$ transitions for projecting away the $y$ component, $O(n)$ transitions for the determinized automaton, giving a total of $O(n)$ transitions in addition to the number of transitions to create $M_{n-1}$.

Putting everything together and taking minimization steps into account, we get the following recursive formula:
\[ T(n) = 
\begin{cases}
    O(n \log n) + T(n/2), & \text{if $n \bmod 2 = 0$}; \\
    O(n \log n) + T(n-1), & \text{if $n \bmod 2 = 1$}.
\end{cases}
\]
So we have $T(n) = O(n\log^2 n)$.    
\end{proof}

\begin{theorem} \label{theorem:n[x]+c=[z]-time}
    Let $n \geq 1$, $0 \leq c < n$ be fixed constants. The automaton recognizing the relation $n[x]_k + c = [z]_k$ with msd-first input can be computed in $O( n\log (n) \log(n \log (n) ) )$ time.
\end{theorem}

\begin{proof}
Now we want to create an automaton $M_{n, c}$ recognizing $n[x]_k + c = [z]_k$.

Let $M_n$ be the automaton from Theorem \ref{theorem:y=n[x]-time} for constant $n$ and let $M_c$ be the automaton from Theorem \ref{theorem:[x]+c=[z]-time} for constant $c$.
We have $M_n(x, y)=1$ if and only if $n[x]_k = [y]_k$ and $M_c(y, z)=1$ if and only if $[y]_k + c = [z]_k$. So we can use the following expression to create $M_{n, c}$:
\[
\exists y \  M_n(x, y)  \AND  M_c(y, z).
\]

We first apply the product construction to $M_n$ and $M_c$. The resulting automaton $M'_{n,c}$ has $O(n \log c) = O(n \log n)$ states.

Consider an input $x \times y \times z$ for $M'_{n, c}$ leading to some state $[q, p]$ where $q$ is from $M_n$ and $p$ is from $M_c$. By the construction of $M'_{n,c}$,  we know that $q=[y]_k - n[x]_k$ and $p=[z]_k - [y]_k$. So we have $q+p = [z]_k - n[x]_k$.
Therefore, to each state in the automaton $M'_{n, c}$ reachable on $x \times y \times z$ a difference $d = [z]_k - n[x]_k$ can be attributed. Furthermore, since $q \in \{0, \ldots, n-1\}$ and $p \in \{0, \ldots, c\}$, the attributed differences are in the range $[0, n+c-1]$. If we choose a difference $d \in [0, n+c-1]$, then we can exactly tell what combinations of $[q, p]$ correspond to $d$ (if any) and no $[q, p]$ corresponds to more than one $d$.

Consider all the states $[0, d], [1, d-1], \ldots, [d, 0]$ corresponding to  the difference $d$, assuming all these states exist in $M'_{n, c}$. If one of these states $[q, d-q]$ is reached on input $x \times y \times z$, then the state $[q+i, d-q-i]$ is reached on input $x \times ([y]_k+i) \times z$. Recall that we established we create automata so that they have the correct output regardless of the number of leading zeros read.

Next, we project away the $y$ component from $M'_{n, c}$.  The resulting NFA is determinized by subset construction to get a DFA.
Consider a state from the DFA reachable on $x \times z$. This state consists of some $[q, p]$ from the NFA such that $q+p = [z]_k - n[x]_k$. Based on the discussion above, there are $d=O(n+c)$ states in the DFA. So the DFA has $O(n+c) = O(n)$ states.

From the proof of Theorem \ref{theorem:n[x]+c=[y]-sc}, we know the automaton recognizing $n[x]_k + c= [z]_k$ on input $x \times z$ for $0 \leq c < n$ only needs to keep track of a difference $[z]_k - n[x]_k$ in the range $[0, n-1]$. So after some minimization, the DFA $M_{n, c}=(Q, \Sigma^2_k, \delta, q_0, A)$ can be formally defined as follows.
\begin{align*}
Q &= \{ d \in \mathbb{N}  \suchthat d \in [0, n-1] \}, \\
\delta(d, [a, e]) &= kd + e - na,
 \\
q_0 &= 0, \\
A &= \{c\}.
\end{align*}
and $M_{n, c}$ has $O(n)$ states.
Taking into account the minimization steps and the time spent on creating $M_n$ and $M_c$ based on the proof of Theorems \ref{theorem:y=n[x]-time} and \ref{theorem:[x]+c=[z]-time} and considering $c=O(n)$, the time spent on creating $M_{n,c}$ is $O( n\log (n) \log(n \log (n) ) )$.
\end{proof}

\begin{theorem} \label{theorem:h(ni+c)-time}
    Let $n \geq 1$, $0 \leq c < n$ be fixed constants. For msd-first input, given an $m$-state DFAO generating $(h(i))_{i \geq 0}$,
    we can compute
    the DFAO generating $(h(ni+c))_{i \geq 0}$ in  $O(  n\log n \log(n \log n )+ nm^2 \log(nm^2)  )$ time.
\end{theorem}

\begin{proof}
    We create the DFAO for $(h(ni+c))_{i \geq 0}$. 
\sloppy
 We are given a DFAO $M_h=(Q_h, \Sigma_k, \delta_h, q_{0,h}, \Delta_h, \tau_h)$ for $(h(i))_{i \geq 0}$ with $m$ states and we have $M_h(y) = h([y]_k)$.
Let $M_{n, c}$ be the automaton from the proof of Theorem \ref{theorem:n[x]+c=[z]-time} where $M_{n,c}(x, y)$ is $1$ if and only if $n[x]_k + c = [y]_k$.

We first use the product construction for $M_{n, c}$ and $M_h$ such that the  $M_{n,c}$ part processes the input $x \times y$ and the $M_h$ part processes the input $y$ to get the automaton $M'$. Consider a state $[q, p]$ reachable on $x\times y$ in $M'$ where $q$ is from $M_{n, c}$ and $p$ is from $M_h$. We have $q = [y]_k -n[x]_k \in [0, n-1]$ and $p=\delta_h(q_{0,h}, y)$. This automaton has $nm$ states.

Next, we project the $y$ component away and use subset construction to construct a DFAO. On input $x$ the automaton reaches a state consisting of some pairs $[q, p]$ where $q \in [0, n-1]$ and $p = \delta_h(q_{0, h}, (n[x]_k+q)_k)$. The $q$ are all values in the range $[0, n-1]$ and the $p$ values are length-$n$ subwords of $(h'(i))_{i \geq0}$. Furthermore, the output of the state is $\tau_h(p)$ where $p$ is from the state $[c, p]$ of $M'$. So based on Theorem \ref{theorem:subword-complexity}, there are $O(nm^2)$ states in the automaton.

Taking into account the minimization steps and the time spent on creating $M_{n, c}$ based on the proof of Theorem \ref{theorem:n[x]+c=[z]-time}, creating $M$ given $M_h$ takes $O(  n\log n \log(n \log n ) + nm^2 \log(nm^2)  )$ time.
\end{proof}

\begin{theorem}
    For msd-first input, given the $m$-state DFAO generating $(h(i))_{i \geq 0}$, the DFAO generating $(h(i+j))_{i, j \geq 0}$ is computed in $O(m^2 \log m)$ time.
\end{theorem}
\begin{proof}
Let $M_h=(Q_h, \Sigma_k, \delta_h, q_{0,h}, \Delta_h, \tau_h)$ be the $m$-state DFAO generating the sequence $(h(i))_{i \geq 0}$; that is,  $M_h(z) = h([z]_k)$.

We first apply the product construction to $M_{\add}$ and $M_h$ such that the $M_\add$ part works on input $x\times y \times z$ and the $M_h$ part works on input $z$ to get the automaton $M'$. Consider a state $[q, p]$ reachable on $x \times y \times z$ in $M'$ where $q$ is from $M_{\add}$ and $p$ is from $M_h$. We have $q=[z]_k-([x]_k + [y]_k)$ and $p=\delta_h(q_{0, h}, z)$. This automaton has $2m$ states.

Next, we project the $z$ component away and use the subset construction to construct a DFAO.  On input $x \times y $, the DFAO reaches a state consisting of some $[q, p]$ where $q \in \{0, 1\}$ and $p = \delta_h(q_{0, h}, (q + [x]_k + [y]_k)_k)$. The $p$ values are length-$2$ subwords of $(h'(i))_{i \geq 0}$. So the DFAO has $O(m^2)$ states.

Taking minimization steps into account, creating $M$ given $M_h$ takes $O(m^2 \log m)$ time.
\end{proof}

\section{Open Problems} \label{section:open-problems}

Here we introduce some other open problems.

Let $(h(i))_{i \geq 0}$ be a $k$-automatic sequence generated by an automaton $M$ of $m$ states. Recall that $M(x)$ is the output of $M$ on input word $x$.
Let $M'$ be an automaton such that on input $x \times y \times z$ the output is $1$ if $h([x])\cdots h([x]+[z]-1) = h([y]) \cdots  h([y]+[z]-1)$, and $0$ otherwise. In other words, the output is $1$ if and only if the length-$[z]$ subwords of the sequence $(h(i))_{i \geq 0}$ starting at positions $[x]$ and $[y]$ are the same. 
We can use either of the following first-order logical expressions to compute $M'$ by an interpretation of B\"uchi arithmetic similar to what we did in Section \ref{section:Buchi}.
\begin{align*}
    &\forall u, \ [u] < [z] \Longrightarrow M(([x]+[u])_k) = M(([y]+[u])_k), \\
    &\forall u, v, \ ([u] \geq [x] \AND [u] < [x] + [z] \AND [u] + [y] = [v] + [x]) \Longrightarrow M(u) = M(v). 
\end{align*}

The state complexity and the time complexity of algorithmically creating $M'$ can be studied similar to the other operations in this paper.
The automaton for $[u] \geq [x] \AND [u] < [x] + [z] \AND [u] + [y] = [v] + [x]$ has $9$ states. So by a simple analysis, for the first one, we can prove an upper bound of $2^{O(m^4)}$ states and for the second
$2^{9m^2}$. However, both of these bounds are likely to be quite weak.  We state the following questions.

\begin{openproblem}
    What is the best lower and upper bound on the state complexity of $M'$ in terms of $m$?
\end{openproblem}

\begin{openproblem}
    How much time is required to compute $M'$ by  an interpretation of B\"uchi arithmetic?
\end{openproblem}

\section{Conclusion}
In this paper we studied various automata recognizing basic relations or carrying out operations on automatic sequences where the automata have input in base $k$. Furthermore, we answered an open question of Zantema and Bosma \cite{zantema_complexity_2022} in Theorem \ref{theorem:h(ni+c)-msd-sc} and Corollary \ref{corollary:h(ni+c)-msd-sc}.

In two forthcoming papers, we address similar topics with inputs in the Fibonacci numeration system \cite{moradi_complexity_2026_2,moradi_state_2026}.

\bibliographystyle{plain}
\bibliography{references}

\end{document}